\documentclass[12pt]{article}
\usepackage[left=1in,right=1in,top=1in,bottom=1in]{geometry}
\usepackage{attrib,amsthm,amsmath,txfonts}
\usepackage{authblk,setspace}

\newtheorem{lemma}{Lemma}
\newtheorem{proposition}{Proposition}
\newtheorem{corollary}{Corollary}
\newtheorem{theorem}{Theorem}
\theoremstyle{definition}
\newtheorem{definition}{Definition}
\DeclareMathOperator{\tr}{tr}
\DeclareMathOperator{\Span}{span}
\DeclareMathOperator{\Int}{int}

\begin{document}
\title{Global Spacetime Similarity}
\author{Samuel C. Fletcher 
}                     
\affil{ 
University of Minnesota, Twin Cities,
USA
}
%
%
%
\maketitle
\doublespacing
\begin{abstract}
There are two classes of topologies most often placed on the space of Lorentz metrics on a fixed manifold. As I interpret a complaint of R.~Geroch [\textit{Relativity}, 259 (1970); \textit{Gen. Rel. Grav.}, 2, 61 (1971)], however, neither of these standard classes correctly captures a notion of global spacetime similarity. In particular, Geroch presents examples to illustrate that one, the compact-open topologies, in general seems to be too coarse, while another, the open (Whitney) topologies, in general seems to be too fine.  After elaborating further the mathematical and physical reasons for these failures, I then construct a topology that succeeds in capturing a notion of global spacetime similarity and investigate some of its mathematical and physical properties.
\end{abstract}
%


\section{Introduction}
A topology on a collection of mathematical objects formalizes a precise notion of similarity amongst them.  It determines which sequences of such objects converge, which parameterized families vary continuously, and which properties of subcollections are generic or are stable under sufficiently small perturbations.  When these objects are spacetimes, examining which limits of sequences and families of spacetimes converge yields both a way to construct new ones and further understanding of the relationships between them, such as how Newtonian spacetimes can be understood as certain limits of relativistic spacetimes \cite{Fletcher14a}.  A topology on the collection of spacetimes also provides a framework for addressing questions of the genericness and stability, hence physical significance, of acausalities and singularities \cite{Geroch70,Geroch71,Fletcher14a,Hawking71,Lerner73}.

Of course, there are always many topologies from which to choose, each capturing a slightly different type of similarity.  Typically, though, one does not begin an investigation with a predetermined topology in hand, but rather with certain intuitions, examples, and other desiderata a desired notion of similarity should respect.  One then attempts to formalize that notion with a choice of topology.  But when such topologies are introduced, insufficient attention is usually paid to justifying how they respect the relevant desiderata, that is, how they really do capture the relevant notion of similarity.  Ignoring this demand can lead to counterintuitive and misleading conclusions, such as it being generic for a spacetime to contain closed timelike curves \cite[Prop.~4]{Fletcher14b}.

One notable exception to this omission has been Geroch \cite{Geroch70,Geroch71}, who argues that no instance of either of the two most commonly used topologies, the \textit{compact-open} and \textit{open} topologies, meets the demands on convergence and continuity imposed by three simple examples.  
(I follow here the terminology of Hawking \cite{Hawking71} and Hawking and Ellis \cite{HE73}. Geroch \cite{Geroch70,Geroch71} calls the compact-open and open topologies the \textit{coarse} and \textit{fine} topologies, respectively.  The open topologies are sometimes called the \textit{Whitney} topologies by mathematicians \cite{GG73}. Caution: Hawking \cite{Hawking71} defines a class of fine topologies that are not in general homeomorphic to the open topologies.)
In the next section I discuss his examples and interpret his desiderata as a demand for a topology that captures a notion of global similarity analogous to the topology of uniform convergence for real functions of a single variable.  In this light the unsuitability of the compact-open topologies is manifest, as they correspond to topologies of uniform convergence on compacta. (I assume the fixed spacetime manifold under consideration is non-compact, as most cases of interest are.  Otherwise the compact-open and open topologies coincide.)  The failure of the open topologies is more subtle, however, since it straightforwardly \textit{is} a topology of uniform convergence in the mathematical sense.  The key to comprehending this failure is the algebraic structure of the Lorentz metrics, understood as a subset of the twice-covariant smooth tensor fields on a manifold, which together form a real vector space of uncountable dimension under pointwise addition.  My diagnosis is that the force of Geroch's examples comes from their implicit demand to make vector addition and scalar multiplication in this space continuous, a demand that the open topologies do not in general satisfy. Indeed, the topology of uniform convergence on an infinite-dimensional vector space is not necessarily compatible with the space's linear structure.

Next, I articulate a sense in which these demands cannot be jointly met.  That is, I show that there is no topology that both satisfies Geroch's desiderata and is compatible with the linear structure of the Lorentz metrics (as a subset of the twice-covariant smooth tensor fields).  The heart of this incompatibility is the fact that every topology compatible with the linear structure of a space must be translation-invariant: roughly put, the neighborhood structure must look the same at every point.  I shall argue that, properly interpreted, one of Geroch's examples requires that whether a ``perturbation" (in a sense that I make precise) of a certain Lorentz metric is ``small" or ``large" must depend on that metric in such a way as to make the same perturbation small for some metrics and large for others.  But if the neighborhood structure at every point is the same, then a perturbation that is small for one metric must be small for all others.  Thus if the twice-covariant smooth tensor fields on a manifold are given any translation-invariant topology, the subspace topology it induces on the Lorentz metrics cannot satisfy Geroch's desiderata.

Fortunately there \textit{is} a way to satisfy Geroch's desiderata if one weakens the requirement of compatibility slightly.  To do so, I construct the \textit{global topology} on the Lorentz metrics, which divides them into an (uncountably) infinite number of uniform components, each corresponding to different ``asymptotic behavior" (in a sense that I make precise).  This topology does makes scalar multiplication continuous, but vector addition will not be so in general.  However, through translation the subspace topology on each uniform component induces a topology on the whole space of twice-covariant smooth tensor fields that \textit{does} make vector addition everywhere continuous, hence \textit{does} make the space a topological vector space.  In other words, ``within" each uniform component the linear operations are continuous.  The construction of the global topologies may be of mathematical interest on its own, for it requires almost no details specific to general relativity.  Indeed, it can be used to construct a global topology on the cross-sections of any seminorm bundle, i.e., a smoothly varying seminorm on the fibers of some vector bundle.  For instance, each of the temporal and spatial metrics with which Newtonian spacetimes are equipped can be understood as smoothly varying seminorms on the tangent bundle of the underlying manifold \cite[p.~250--4]{Malament12}. (There is only an indirect sense in which the Newtonian spatial metric, as a symmetric $(2,0)$-tensor field, assigns a length to the spacelike component of a vector at a point, but the length it does assign is unique \cite[Prop.~4.1.1]{Malament12}.) In addition to these technical properties, the global topologies also have a natural physical interpretation. They encode similarity amongst spacetimes as similarity of values of globally defined quantities that are measurable, in a precise sense, with bounded precision.  This contrasts with the analogous characterization of the open topologies, which also encode similarity of values of globally defined observables, but allow for measurement apparatuses of arbitrary precision.

Before concluding with some open questions concerning invariant and geometrical properties of the global topologies, I discuss their application to our understanding of the stability and genericness of global properties of spacetimes.  How do results proved using the open topology fare with the global topologies?  I show that for virtually any theorem concerning the stability and genericness of \textit{conformally invariant} properties with respect to an open topology, there is a counterpart with respect to the analogous global topology that also holds.  In particular, the counterpart of Hawking's theorem \cite{Hawking69}, expressing the equivalence of stable causality and the existence of a global time function, holds with respect to the global topology.  But it remains an open question whether and which theorems concerning the stability of other properties, such as geodesic (in)completeness, have counterparts that hold in the global topologies.

\section{Geroch's Complaint} \label{sec:complaint}

In what follows, $M$ will always denote some smooth, paracompact, Hausdorff manifold of dimension $n \geq 2$, $g_{ab}$ some smooth Lorentz metric on $M$, and $h^{ab}$ some (inverse of a) smooth Riemannian metric on $M$, where the indices are abstract.  When these metric fields appear in contexts where their indices will not be contracted, I will often relax the notation by dropping the indices.    The collections of all such Lorentzian and Riemannian metrics will be denoted $L(M)$ and $\mathcal{R}(M)$, respectively.

\subsection{Preliminaries}
The basic idea for topologies on the Lorentzian $g$ considered in this paper is that one can divide the task of measuring the similarity of each pair $g,g'$ into two parts: first, that of encoding their relevant differences at each point of $M$ into a real number in some systematic way, and second, evaluating the variability of the resulting scalar field in some way over regions of $M$.  The differences between the topologies considered arise ultimately from different choices of how to implement these two tasks.  But one feature they all share is the use of the Riemannian $h$ to define a \textit{norm} in every fiber of each tensor bundle $T^r_s M \to M$, i.e, a real function $\nu_p$ of the $(r,s)$-tensors $K^{a_1 \cdots a_r}_{b_1 \cdots b_s}$ at $p \in M$ with the following properties:  
\begin{description}
\item[Homogeneous] For any $\alpha \in \mathbb{R}$, $\nu_p(\alpha K^{a_1 \cdots a_r}_{b_1 \cdots b_s}) = |\alpha| \nu_p(K^{a_1 \cdots a_r}_{b_1 \cdots b_s})$.
\item[Subadditive] For any $(r,s)$-tensor $\hat{K}^{a_1 \cdots a_r}_{b_1 \cdots b_s}$ at $p$, $\nu_p(K^{a_1 \cdots a_r}_{b_1 \cdots b_s} + \hat{K}^{a_1 \cdots a_r}_{b_1 \cdots b_s}) \leq \nu_p(K^{a_1 \cdots a_r}_{b_1 \cdots b_s}) + \nu_p(\hat{K}^{a_1 \cdots a_r}_{b_1 \cdots b_s})$
\item[Separating] If  $\nu_p(K^{a_1 \cdots a_r}_{b_1 \cdots b_s}) = 0$, then $K^{a_1 \cdots a_r}_{b_1 \cdots b_s}$ is the zero tensor.
\end{description}
It will turn out while the choice of norm will not matter substantively to the results in the sequel, the Frobenius norm offers certain computational advantages. (This choice is also the unique $l^p$ norm on a tensor space that allows one to construct an inner product from the norm according to the polarization identity, but this fact will not play a role in this paper.)
\begin{definition}
For any $(r,s)$-tensor $K^{a_1 \cdots a_r}_{b_1 \cdots b_s}$ at $p \in M$, define the \textit{$h$-fiber norm} of $K^{a_1 \cdots a_r}_{b_1 \cdots b_s}$ as
\begin{equation}
|K^{a_1 \cdots a_r}_{b_1 \cdots b_s}|_h = |K^{a_1 \cdots a_r}_{b_1 \cdots b_s} K^{c_1 \cdots c_r}_{d_1 \cdots d_s} h_{a_1 c_1} \cdots h_{a_r c_r} h^{b_1 d_1} \cdots h^{b_s d_s} |^{1/2},
\end{equation}
for $r,s > 0$, and with no copies of $h$ and its inverse when, respectively, $r=0$ and $s=0$.
\end{definition}
Note that the $h$-fiber norm of a scalar is just the absolute value of that scalar, hence independent of the choice of $h$.  Now, because $h$ is positive definite, one can choose a (co)basis for the (co)tangent space at $p$ in which the matrix representation of $h_{|p}$ is the identity.  It then follows immediately from the definition of the $h$-fiber norm that:
\begin{proposition}\label{prop:fiber norm}
For every $p \in M$ and each non-negative $r$ and $s$, the $h$-fiber norm is a norm on the $(r,s)$-tensors at $p$, namely the Frobenius norm with respect to the basis in which $h$ is the identity.
\end{proposition}

An example that will be central in the sequel will be the $h$-fiber norm of the difference of two Lorentz metrics $g$ and $g'$:
\[
|g-g'|_h = [(g_{ab}-g'_{ab})(g_{cd}-g'_{cd})h^{ac}h^{bd}]^{1/2}
\]
at each $p \in M$.  This (and indeed any $h$-fiber norm) defines a scalar field on $M$ whose evaluation is the concern of the second aforementioned task in measuring the similarity between $g$ and $g'$.  All the topologies considered also share a common approach to this task.  
\begin{definition}
Let $|\cdot|_h$ be some $h$-fiber norm on $M$ and let $S \subseteq M$ be nonempty.  Then for any $(r,s)$-tensor field $K^{a_1 \cdots a_r}_{b_1 \cdots b_s}$, define its \textit{$(h,S)$-uniform norm} as
\begin{equation}
\| K^{a_1 \cdots a_r}_{b_1 \cdots b_s} \|_{h,S} = \sup_{S} |K^{a_1 \cdots a_r}_{b_1 \cdots b_s}|_h.
\end{equation}
\end{definition}
More generally, one could consider an $(h,S)$-$L^p$ norm, 
\[
\left(\int_S \left( |K^{a_1 \cdots a_r}_{b_1 \cdots b_s}|_h \right)^p dV \right)^{1/p},
\]
where $V$ is the volume measure determined by $h$, but I will restrict attention to the uniform ($p = \infty$) case here. (The case $p=2$ is used in applications to the Cauchy problem in general relativity \cite[Ch.~7.4]{HE73}. Technically, the $p = \infty$ case uses the essential supremum, but this is equal to the supremum for continuous functions.)

It will be absolutely central to the sequel that, in general, the $(h,S)$-uniform norm for some tensor fields is infinite.
\begin{proposition}
For every nonnegative integer $r$ and $s$, each $(h,S)$-uniform norm is an \emph{extended norm} on the vector space of continuous $(r,s)$-tensor fields on $M$, i.e., it is a homogeneous, subadditive, and separating function whose range is $[0,\infty]$.
\end{proposition}
\begin{proof}
Immediate from proposition \ref{prop:fiber norm} and the subadditivity of the supremum.
\end{proof}

\subsection{The Compact-Open and Open Topologies}

Using the above framework, all topologies under consideration render two tensor fields similar, roughly speaking, when their values and partial derivatives up to $k$th order are sufficiently similar at points of some $S \subseteq M$.  This can be captured for smooth $(r,s)$-tensor fields $K$ through an analog to the $\epsilon$-ball notion of similarity familiar from metric spaces, with the relevant distance function defined from an $(h,S)$-norm:
\begin{equation}\label{eq:basis}
B^k(K,\epsilon;h,S) = \{K' \in \Gamma^r_s : \forall j \leq k, \| \nabla^{(k)} K - \nabla^{(k)} K' \|_{h,S} < \epsilon \},
\end{equation}
where $\Gamma^r_s$ denotes the collections of smooth $(r,s)$-tensors on $M$, $\nabla^{(k)} K$ the field $\nabla_{c_1} \cdots \nabla_{c_k} K^{a_1 \cdots a_r}_{b_1 \cdots b_s}$, and $\nabla$ the Levi-Civita derivative operator compatible with $h$.  The topologies that will be under consideration here are all determined by a suitable collection of such $\epsilon$-balls as a \textit{subbasis}, i.e., a topology's open sets will be generated through finite intersection and arbitrary unions of these $\epsilon$-balls. (In some circumstances they form a \textit{basis}, which generates the topology only through finite intersection.)  In particular, a $C^k$ topology on $\Gamma^r_s$ generated in this way is determined by the appropriate choice of quadruples $(K,\epsilon;h,S)$. Letting $\mathcal{P}(M)$ denote the power set of $M$ and $\mathcal{R}(M)$ the smooth Riemannian metrics on $M$, these collections may be called \textit{definitional subbases}:
\begin{definition}
A collection of quadruples $\Xi \subseteq  \Gamma^r_s \times (0,\infty) \times \mathcal{R}(M) \times \mathcal{P}(M)$ is a \textit{definitional subbasis} for a topology on $\Gamma^r_s$ when $\bigcup \pi_1[\Xi] = \Gamma^r_s$, where $\pi_1 : (K,\epsilon;h,S) \mapsto K$ is the projection onto the first component.
\end{definition}
Indeed, the $C^k$ compact-open and open topologies can be defined this way.  Letting $\mathcal{C}$ denote the collection of all compact subsets of $M$,

\begin{definition}
The definitional subbasis for the $C^k$ \textit{compact-open} topology on $\Gamma^r_s$ is $\Xi_{CO} = \Gamma^r_s \times (0,\infty) \times \mathcal{R}(M) \times \mathcal{C}$.
\end{definition}
\begin{definition}
The definitional subbasis for the $C^k$ \textit{open} topology on $\Gamma^r_s$ is $\Xi_O = \Gamma^r_s \times (0,\infty) \times \mathcal{R}(M) \times \{M\}$. (This topology is called \textit{open} since replacing $\{M\}$ with the collection of all open subsets of $M$ generates the same topology.)
\end{definition}
The $C^\infty$ compact-open topology on $\Gamma^r_s$ is generated from the union of the $C^k$ compact-open topologies for all $k$, and similarly for the open topologies. Also, if $M$ is compact, the $C^k$ compact-open topology coincides with the $C^k$ open topology.  Otherwise, the $C^k$ compact-open topology is strictly coarser than the $C^k$ open topology, i.e, the system of open sets of the former is strictly contained in the system of the latter.  For $j < k$, the $C^j$ compact-open topology is also strictly coarser than the $C^k$ compact-open topology, and similarly for the open topologies.  Finally, the compact-open and open topologies on the smooth Lorentz metrics on $M$, denoted $L(M)$, are just the relevant subspace topologies on $\Gamma^0_2$.

\subsection{Geroch's Examples}

Geroch \cite{Geroch70,Geroch71} uses three examples---two sequences and a single one-parameter family of Lorentz metrics---to illustrate his claim that each of the $C^k$ compact-open topologies is too coarse while each of the $C^k$ open topologies is too fine.  Of course, a topology can only be too coarse or fine for a particular purpose.  There are good reasons to believe that there is no canonical topology on the space of smooth Lorentz metrics, so particular choices of topology must be well adapted to the constraints on a notion of similarity relevant for the context at hand \cite{Fletcher14b}.  I shall argue that one can thus interpret these examples as illustrating the inadequacy of these topologies for encoding a notion of \textit{global} similarity.

Let $t,x,y,z$ be scalar coordinate fields on $\mathbb{R}^4$.  Geroch's two sequences of metrics on this manifold are
\begin{eqnarray}
\stackrel{m}{g}_{ab} &= \left( 1 + \frac{m}{1+(x-m)^2} \right)(d_a t)(d_b t) - (d_a x)(d_b x) - (d_a y)(d_b y) - (d_a z)(d_b z), \label{eq:example1a} \\
\stackrel{m}{g}{}'_{ab} &= \left( 1 + \frac{1}{m^2 + x^2 + y^2 + z^2} \right) (d_a t) (d_b t) - (d_a x) (d_b x) - (d_a y) (d_b y) - (d_a z) (d_b z), \label{eq:example2}
\end{eqnarray}
where $d$ is the exterior derivative. (The formula for the first term of eq.~\ref{eq:example1a} is garbled in \cite{Geroch71}, but appears without error in \cite[p.~280]{Geroch70}.)  His one-parameter family is simply
\begin{equation}
\{ \lambda g_{ab} : \lambda > 0 \}, \label{eq:example2b}
\end{equation}
for some arbitrary fixed Lorentz metric $g_{ab}$ on any non-compact $M$.  Regarding eq.~\ref{eq:example1a}, Geroch writes that ``The `bump' in the metrics becomes larger as it recedes to infinity,'' but the ``sequence \emph{does} approach Minkowski space in the [compact-open] topology (because the metrics become Minkowskian in every compact set).'' However, ``[i]ntuitively, we would not think of this sequence as approaching Minkowski space'' \cite[p.~71]{Geroch71}.  

On the other hand, with respect to any open topology the sequence defined by eq.~\ref{eq:example2} does not converge to the Minkowski metric even though ``[t]he `bump' in the metrics remains centered at the origin and decreases in amplitude'' to zero \cite[p.~280]{Geroch70}.  Perhaps more surprisingly, the family defined in eq.~\ref{eq:example2b} does not trace out a continuous curve in the open topology.  In fact, it is everywhere discontinuous since the induced subspace topology on this set is discrete \cite[p.~71]{Geroch71}.  This is particularly striking since each member of this family can represent precisely the same physical spacetimes \cite{Fletcher18}. (One might thus object there is no problem here, since after taking the quotient of the space of Lorentz metrics by the isometry relation the family becomes a point.  Such a strategy does not help with the non-convergence of eq.~\ref{eq:example2}, however, since no two elements of that sequence are isometric.)

But the compact-open topologies \textit{do} determine that the sequence defined in eq.~\ref{eq:example2} converges to Minkowski spacetime and that the family defined in eq.~\ref{eq:example2b} \textit{is} continuous, while according to the open topologies the sequence defined in eq.~\ref{eq:example1a} does \textit{not} converge.  Thus the compact-open and open topologies are obverses on these examples, each intuitively ruling wrongly on the ones the other rules rightly.  The former is too coarse because it permits too many sequences to converge, while the latter is too fine because it permits too few sequences to converge (and too few continuous families).

\subsection{Diagnosis and Desiderata}

But what, more precisely, is undergirding these intuitions?  In the first place, as Geroch alludes, the compact-open topology coincides with the topology of compact convergence---that is, a sequence of metrics $\stackrel{m}{g} \to g$ on $M$ just when its elements converge uniformly on each compact $C \subseteq M$ \cite[p.~284, Theorem 43.6a]{Willard70}.  This is similar to the sense in which the $k^{th}$-order Taylor expansion of a real function such as $\sin(x)$ converges to it as $k \to \infty$.  (For any particular finite-order expansion, one can find a sufficiently large $x$ such that the expansion, evaluated at $x$, differs from $\sin(x)$ by as much as one wishes.  But if one fixes some compact neighborhood region of $\mathbb{R}$, then the Taylor series converges uniformly on that neighborhood.)  It seems, then, that Geroch is searching for a topology that compares two metrics across the \textit{whole} (non-compact) spacetime, not just on compacta.  This would generate the analog to \textit{uniform convergence} for real functions.

Yet the open topologies \textit{are} topologies of uniform convergence in the usual way the latter are defined. (This is for the same reason as the fact that the compact-open topology coincides with the topology of uniform convergence on compacta.)   In the present context, however,  where the space of functions to topologize are smooth sections of a vector bundle over a non-compact manifold, Geroch's examples are evidence that the criteria for convergence and continuity that the open topologies require is much too strong.  The following propositions (respectively adapted from \cite[pp.~43--4]{GG73} and \cite[Prop.~2]{Fletcher14b}) make this precise.
\begin{proposition}\label{prop:open convergence}
Let $K,\{ \stackrel{m}{K} \}_{m \in \mathbb{N}}$ be tensor fields on $M$.  Then $\stackrel{m}{K} \to K$ in the $C^k$ open topology iff there is a compact $C \subseteq M$ such that:
\begin{enumerate}
	\item for sufficiently large $m$, $\stackrel{m}{K}_{|M-C} = K_{|M-C}$; and
	\item $\stackrel{m}{K}_{|\Int(C)} \to K_{|\Int(C)}$ in the $C^k$ open topology.
\end{enumerate}
\end{proposition}
\begin{proposition}
A family $\{\stackrel{\lambda}{K}\}_{\lambda \in I}$ of tensor fields on $M$ is continuous in the $C^k$ open topology only if for every $\lambda_1,\lambda_2 \in I$, there is some compact $C \subseteq M$ such that $\stackrel{\lambda_1}{K}_{|M-C} = \stackrel{\lambda_2}{K}_{|M-C}$.
\end{proposition}
In light of these propositions, the open topologies would in a sense be appropriate if the sections under consideration had only compact support.  But this restriction is clearly inappropriate for Lorentz metrics (as they are everywhere non-degenerate). (Hawking has expressed doubt that another class of topologies, which he calls the fine topologies and which requires two metrics (and their derivatives to order $k$) to be equal outside of a compact set to be similar, is as ``physical'' as the open topology ``since to establish that two metrics actually coincide outside some compact set would require an exact measurement which is not physically possible'' \cite[p.~397]{Hawking71}.  But this is exactly what is required to establish that a sequence of metrics converges in one of the open topologies, so if one accepts this line of reasoning, one must also reject the open topologies on the same grounds.)

One way to capture the essential insight of Geroch's latter two examples is that in addition to requiring a topology that is ``sensitive" to differences across all of $M$---that is, demanding that the last component of each member of its defining subbasis should be $M$---one should also require that the topology respect the linear structure of the Lorentz metrics.  Each class of $(r,s)$-tensor fields (indeed, any vector bundle) on $M$ forms a real vector space, with vector addition defined pointwise on $M$ and scalar multiplication acting fiberwise. (They also form a module over the ring of continuous real scalar fields on $M$, but this fact will play no role in what follows.)  As a subset of the smooth $(0,2)$-tensor fields, the Lorentz metrics inherit this algebraic structure. (The Lorentz metrics themselves are not a vector space, as they are not closed under scalar multiplication (consider negative scalars) or vector addition.  Nevertheless, what's important is that these operations be continuous on the domain of Lorentz metrics on which they \textit{are} defined.)  Demanding the continuity of the one-parameter family exhibited in eq.~\ref{eq:example2b} matches exactly the requirement that scalar multiplication be continuous; the sequence defined in eq.~\ref{eq:example2} suggests the requirement that vector addition be continuous as well---at least for ``small" vectors.  (See the discussion at the end of this section motivating the $C^k$ global topologies constructed in the next section.)

There is no paradox that the open topologies do not make the vector space operations continuous, for the topology of uniform convergence for a vector space need not be compatible with its linear structure, i.e., make it into a \textit{topological} vector space, as illustrated through the following proposition (adapted from \cite[Theorem 26.29, p.~705]{Schechter97}):
\begin{proposition}\label{prop:tvs}
A topology on a vector space $V$ makes it into a (locally convex) topological vector space if and only if the topology
\begin{enumerate}
\item is translation-invariant, i.e., each translation map $v \mapsto u + v$ for a fixed $u \in V$ is a homeomorphism of $V$, and
\item has a 0-neighborhood base consisting of $\epsilon$-balls generated from some collection of seminorms, i.e., a collection of homogeneous and subadditive (but not necessarily separating) functions.
\end{enumerate}
\end{proposition}
Both the open and compact-open topologies are translation-invariant.  But the open topologies have 0-neighborhoods bases generated from the collection of the $(h,M)$-uniform norms, which are \textit{extended} norms (hence extended seminorms)---they assign infinite length to some tensor fields.  Thus they are not compatible with the linear structure of $\Gamma^r_s$ for any $r,s > 0$.  The compact-open topologies, meanwhile, have 0-neighborhood bases generated from the $(h,C)$-uniform norms for a compact $C$, which are in fact norms (hence seminorms) because the $h$-fiber norms, as continuous real functions on $M$, are always bounded on compacta.

\section{The Global Topologies}\label{sec:uniform}
\subsection{Motivation: The Problem with Translation-Invariance}

Intuitively, a sequence should converge just when it eventually becomes arbitrarily similar to a limit point.  The open topologies determine that eq.~\ref{eq:example2} fails that criterion essentially because they allow any choice of Riemannian metric $h$ in their definitional subbasis $\Xi_{CO}$.  When $M$ is non-compact, there will always be some $h$-fiber norm that ``blows up'' sufficiently rapidly ``at infinity.'' More precisely, given any two tensors $K$ and $K'$ of the same rank that differ outside of each compact subset of $M$, there will always be some $(h,M)$-uniform norm of their difference that will be infinite.  This explains why the open topology has the convergence condition that it does: only when $K$ and $K'$ differ on a compact set will each $(h,M)$-uniform norm assign a finite value to their difference (by the extreme value theorem).

Now, Geroch characterizes the sequence defined by eq.~\ref{eq:example2} as one where the ``bump'' flattens further and further as the sequence progresses.  The open topologies do not capture this since there are some $h$-fiber norms for which $|\stackrel{m}{g} - \eta|_h$ does not flatten but is in fact unbounded on $\mathbb{R}^4$.  (Here $\eta$ is the Minkowski metric.)  A natural response, which Geroch himself has suggested \cite[p.~288]{Geroch70}, is to restrict the choices of $h$ to be used in the definitional subbasis.  (Such a restriction should also not be \emph{ad hoc}.  One could easily fix a \emph{single} choice of $h$, but this would be completely arbitrary.  Moreover, it would not be an \emph{invariant} topology in the sense articulated below.)  

But before attempting such a response, it will be helpful to reflect on why Geroch's claim that eq.~\ref{eq:example2} should converge seems plausible in the first place.  The difference between the Minkowksi metric and the $m^{th}$ term in the sequence is 
\begin{equation}
\stackrel{m}{g}{}'_{ab} - \eta_{ab} = (m^2 + x^2 + y^2 + z^2)^{-1} (d_a t)(d_b t).
\label{eq:difference}
\end{equation}
This difference will be everywhere small only according to an $h$-fiber norm that is well adapted to the chosen coordinates, which are in turn well adapted to the limit point, the Minkowskian metric.  That is, if
\begin{equation}
	h^{ab} = \left( \frac{\partial}{\partial t} \right)^a \left( \frac{\partial}{\partial t} \right)^b + 
	\left( \frac{\partial}{\partial x} \right)^a \left( \frac{\partial}{\partial x} \right)^b + 
	\left( \frac{\partial}{\partial y} \right)^a \left( \frac{\partial}{\partial y} \right)^b + 
	\left( \frac{\partial}{\partial z} \right)^a \left( \frac{\partial}{\partial z} \right)^b,
\end{equation}
then $| \stackrel{m}{g}{}' - \eta |_h = (m^2 + x^2 + y^2 + z^2)^{-1}$ and $\| \stackrel{m}{g}{}' - \eta \|_{h,M} = m^{-2}$, which vanishes as $m \to \infty$.

But if the example were modified slightly the result would be different.  Consider the sequence
\begin{equation}
\stackrel{m}{g}{}''_{ab} = \left( e^t + \frac{1}{m^2 + x^2 + y^2 + z^2} \right) (d_a t) (d_b t) - (d_a x) (d_b x) - (d_a y) (d_b y) - (d_a z) (d_b z)
\end{equation}
on $\mathbb{R}^4$ with plausible limit point
\begin{equation}
\eta'_{ab} = e^t (d_a t) (d_b t) - (d_a x) (d_b x) - (d_a y) (d_b y) - (d_a z) (d_b z).
\end{equation}
Their difference is the same tensor as that for the analogous case introduced by Geroch, eq.~\ref{eq:difference}.  But the (inverse) Riemannian metric well adapted to $\eta'$ is
\begin{equation}
	h'^{ab} = e^{-t} \left( \frac{\partial}{\partial t} \right)^a \left( \frac{\partial}{\partial t} \right)^b + 
	\left( \frac{\partial}{\partial x} \right)^a \left( \frac{\partial}{\partial x} \right)^b + 
	\left( \frac{\partial}{\partial y} \right)^a \left( \frac{\partial}{\partial y} \right)^b + 
	\left( \frac{\partial}{\partial z} \right)^a \left( \frac{\partial}{\partial z} \right)^b,
\end{equation}
which yields that $| \stackrel{m}{g}{}'' - \eta' |_{h'} = e^{-t}/(m^2 + x^2 + y^2 + z^2)$, which is unbounded.  Thus even though $\stackrel{m}{g}{}' - \eta = \stackrel{m}{g}{}'' - \eta'$, the former sequence is taken to converge while the latter is not.

This observation poses a dilemma for Geroch's suggestion that one should try to find a topology intermediate between the compact-open and open topologies by means of some definitional subbasis $(\Gamma^0_2,\mathbb{R}_+,\mathcal{R}',\{M\})$, where $\mathcal{R}' \subset \mathcal{R}(M)$, that respects the linear structure of $\Gamma^0_2$.  By proposition \ref{prop:tvs}, such a topology must be translation-invariant.  But any translation-invariant topology on $\Gamma^0_2$ must arrive at the same answer as to whether the sequences $\stackrel{m}{g}{}'$ and $\stackrel{m}{g}{}''$ converge.  This is in conflict with the above conclusion that the $\stackrel{m}{g}{}'$ but not the $\stackrel{m}{g}{}''$ should converge.   

I would like to suggest that one can still satisfy a weakened form of Geroch's desiderata by giving up translation invariance and a topology fully compatible with linear structure.  As I adumbrated in the introduction, one can still retain a weakened form of these features, in which the Lorentz metrics are partitioned into components, each of which is itself a subset of a topological vector space.  In other words, ``within" each of these components, scalar multiplication and vector addition are continuous, so Geroch's desiderata are satisfied for ``small" but not ``large" perturbations---those that lie in the same component as as the unperturbed metric.  This is made precise in the next subsection.

\subsection{The Definition and Position of the Global Topologies}

First, one can extend the notion of a fiber norm in the following way.  Given \textit{any} non-degenerate metric $f$ and any $(r,s)$-tensor field $K$ on $M$, define the \textit{$f$-fiber norm of $K$} to be
\begin{equation}
|K^{a_1 \cdots a_r}_{b_1 \cdots b_s}|_f = |K^{a_1 \cdots a_r}_{b_1 \cdots b_s} K^{c_1 \cdots c_r}_{d_1 \cdots d_s} f_{a_1 c_1} \cdots f_{a_r c_r} f^{b_1 d_1} \cdots f^{b_s d_s} |^{1/2},
\end{equation}
and the \textit{$(f,S)$-uniform norm of $K$} for any $S \subseteq M$ to be $\| K \|_{f,S} = \sup_S |K|_f$.  Note that if $K$ is a scalar field then its fiber norm (hence the uniform norm over any set as well) is independent of the choice of $f$.

Next, for any $p \in M$, define $\ker_{r,s}(f_{|p}) = \{ K \in \Gamma^r_s : (|K|_f)_{|p} = 0 \} $ and $\ker_{r,s}(f_{|S}) = \bigcup_{p \in S} \ker_{r,s}(f_{|p})$.  These sets, respectively called the kernels of $f$ at $p$ and $S$, are empty for positive and negative definite metrics but non-empty for metrics of indefinite signature. Finally, let
\begin{equation}\label{eq:K-set}
\mathcal{K}^r_s(f,S) = \Span(\{ K \in \Gamma^r_s - \ker_{r,s}(f_{|S}) : \|K\|_{f,S} < \infty \}).
\end{equation}
This set consists of all the linear combinations of $(r,s)$-tensor fields on $M$ that are nowhere in the kernel of $f$ at $S$ and are bounded with respect to the $(f,S)$-uniform norm.  Finally:
\begin{definition}
Two metrics $f,f'$ are \textit{norm-equivalent on $S$}, denoted $f_{|S} \asymp f'_{|S}$, when $\mathcal{K}^r_s(f,S) = \mathcal{K}^r_s(f',S)$ for all $(r,s)$.
\end{definition}
Metrics that are norm-equivalent on $S$ agree regarding which tensor fields are bounded on $S$.  By definition, norm-equivalence on $S$ is an equivalence relation, hence partitions the metrics on $M$ into equivalence classes $[f_{|S}] = \{f':f_{|S} \asymp f'_{|S} \}$. (One may also partition the metrics into classes that are projectively and affinely equivalent, each of which take on particularly simple forms.  Further, norm-equivalence is refined by projective equivalence, which in turn is refined by affine equivalence.  Nevertheless, I conjecture that replacing norm-equivalence by one of the latter two in the definition of the gloabl topologies (eq.~\ref{eq:global topologies}) yields the same topology, and that no equivalence that strictly refines norm equivalence does so.)  When $S = M$, I shall omit reference to it, writing $[f] = \{f':f_{|M} \asymp f'_{|M} \}$ and calling them (simply) the metrics norm-equivalent to $f$.

Of course, to determine whether two metrics $f,f'$ are norm-equivalent on $S$, one does not need to check that the collections of tensor fields they determine by eq.~\ref{eq:K-set} for each $r,s$ are respectively equal.  It is enough to do so on appropriate subcollections:
\begin{proposition} \label{prop:span}
Let $\mathcal{BK}^r_s(f,S) \subset \mathcal{K}^r_s(f,S)$ denote a collection of tensor fields that spans $\Gamma^r_s$.  Then two metrics $f,f'$ are norm-equivalent on $S$ if and only if  $\mathcal{BK}^r_{s'}(f,S)=\mathcal{BK}^r_{s'}(f',S)$ or $\mathcal{BK}^{r'}_s(f,S)=\mathcal{BK}^{r'}_s(f',S)$ for some $r,s > 0$.
\end{proposition}
%
%
%
%
%
Thus the values that the $(f,S)$-uniform norm assigns to vector fields essentially determines which norm-equivalence class it belongs to. In particular, it is their asymptotic behavior that norm-equivalence encodes.  One can make this precise by adopting the following definition:
\begin{definition}
Let $K,K'$ be two smooth tensor fields and $f$ be a metric.  $K$ and $K'$ are \textit{mutually bounding to order $k$ with respect to $f$ on $S \subseteq M$} when there are constants $c,c' >0$ such that $c|\nabla^{(j)} K|_f \leq |\nabla^{(j)} K'|_f \leq c'|\nabla^{(j)} K|_f$ everywhere on $S$ for each $j \leq k$, where $\nabla$ is the Levi-Civita derivative operator compatible with $f$. (Sometimes the notation $K' \in \Theta(K)$ is used for an analogous notion in computer science, especially the analysis of algorithms, although its origins are in analytic number theory.)
\end{definition}
Note that mutual boundedness is a symmetric relation, for if there are constants $c,c' > 0$ such that $c|\nabla^{(j)} K|_f \leq |\nabla^{(j)} K'|_f \leq c'|\nabla^{(j)} K|_f$ everywhere on some $S$, then $c'{}^{-1}|\nabla^{(j)} K'|_f \leq |\nabla^{(j)} K|_f \leq c^{-1}|\nabla^{(j)} K|_f$ everywhere on $S$.  There is a close connection between norm-equivalence and the mutual boundedness of metrics (anticipated in the notation for norm-equivalence, which has been adapted from \cite[\S1.6]{HW79}):

\begin{theorem} \label{thm:mutually bounding}
Two metrics $f,f'$ are norm-equivalent on some $S \subseteq M$ if and only if they are mutually bounding to order $0$ on $S$ with respect to every metric $f''$.
\end{theorem}
\begin{proof}
Suppose that $f_{|S} \asymp f'_{|S}$ and let $f''$ be arbitrary.  Because $f''$ and $f$ are metrics, the field $\Omega = 1/|f''|_f$ is well-defined on $S$.  Since then $|\Omega f''|_f = 1$ on $S$, $\Omega f'' \in \mathcal{K}^0_2(f,S)$, hence by assumption $\Omega f '' \in \mathcal{K}^0_2(f',S)$.  By definition $\|\Omega f''\|_{f'} \geq |\Omega f''|_{f'} = \Omega|f''|_{f'}$, so $|f''|_{f'} \leq \|\Omega f''\|_{f'} |f''|_f$.  A similar argument returns that $|f''|_f \leq \|\Omega' f''\|_f |f''|_{f'}$ for $\Omega' = 1/|f''|_{f'}$. Combining the two then yields $(\|\Omega' f''\|_f)^{-1}|f''|_f \leq |f''|_{f'} \leq \|\Omega f''\|_{f'} |f''|_f$. Finally, note that $|f''|_f = |f|_{f''}$ and $|f''|_{f'} = |f'|_{f''}$, recalling that $f''$ was arbitrary.

Conversely, suppose that $f,f'$ are mutually bounding to order $0$ with respect to every metric $f''$.  Then by definition there are constants $c,c' > 0$ such that $c|f''|_f \leq |f''|_{f'} \leq c'|f''|_f$.  Hence if $f'' \in \mathcal{K}^0_2(f,S)$, then $f'' \in \mathcal{K}^0_2(f',S)$ and vice versa.  A similar argument yields the analogous conclusion for the inverses of $f$.  Since these together are tensor fields whose basis components in each index span the corresponding (co-)tangent space at every point of $S$, proposition \ref{prop:span} and the definition of norm-equivalence yield that $f_{|S} \asymp f'_{|S}$.
\end{proof}
Two tensor fields that are mutually bounding to order $k$ with respect to some metric need not be mutually bounding to any higher order.  (For instance, consider any field with support on $\mathbb{R}^2$, and the conformally related field obtained through the factor $1+a \sin(x^2)$ for $0<|a|<1$.)  Thus there is a sense in which the above proposition is tight. (It does suggest, though, investigating equivalence relations defined through analogs of theorem \ref{thm:mutually bounding} formed by replacing ``$0$" with some positive integer.  Would such a relation have any interesting physical significance?)

We are now ready to describe the $C^k$ global topologies, using the definitional subbasis 
\begin{equation}
\Xi_G = \{ (g,\epsilon;h,S) \in L(M) \times (0,\infty) \times \mathcal{R}(M) \times \{ M \} : h \asymp g \}.
\label{eq:global topologies}
\end{equation}
(Cf.~the construction of the Banach space $\mathcal{B}_k$ in \cite{LP74}, which however only yields neighborhoods of the Minkowski metric.)
In fact, the neighborhoods $B^k(g,\epsilon;h,M)$ defined by $\Xi_G$ form a (local) basis for the global $C^k$ topology on $L(M)$.  Showing this will be crucial in proving many facts about the global topology, but it requires first a lemma and a preliminary proposition.

\begin{lemma} \label{lemma:bounded}
If $g' \in B^k(g,\epsilon;h,M)$ with $(g,\epsilon;h,M) \in \Xi_G$, then $h \asymp g'$.
\end{lemma}
\begin{proof}
Let $g'$ be given, and recall from the proof of proposition \ref{prop:fiber norm} that, at any $p \in M$, the $h$-fiber norm $(|g-g'|_h)_{|p} = \| (g_{\alpha\beta})_{|p} - (g'_{\alpha\beta})_{|p} \|_F$ is just the Frobenius norm with respect to the basis in which the matrix representation of $h_{|p}$ is the identity.  Thus 
\begin{eqnarray*}
(|g'|_h)_{|p} &=& \| (g'_{\alpha \beta})_{|p} \|_F \\
&\leq&  \| (g_{\alpha \beta})_{|p} \|_F + \| (g'_{\alpha \beta})_{|p} - (g_{\alpha \beta})_{|p} \|_F = (|g|_h)_{|p} + (|g-g'|_h)_{|p}
\end{eqnarray*}
by the triangle inequality.  Since $\| g \|_h < \infty$ and $\| g-g' \|_h < \epsilon$ by hypothesis, we have that $h \asymp g'$.
\end{proof}

Next we need some information about the geometric structure of the Riemannian metrics in each $[f]$.

\begin{proposition} \label{prop:semigroup}
For any smooth, non-degenerate metric $f$, the set $[f] \cap \mathcal{R}(M)$ is a (blunt) convex cone, i.e., if $h_{ab},h'_{ab} \in [f]$ and $c,c'>0$, then $ch_{ab} + c'h'_{ab} \in [f]$.
\end{proposition}
\begin{proof}
It suffices to show that $h_{ab} + h'_{ab} \in [f]$ for any $h_{ab},h'_{ab} \in [f]$.  Recall again from the proof of proposition \ref{prop:fiber norm} that, at any $p \in M$, we can express $(|g|_{h+h'})_{|p} = \| G(h+h',p) \|_F$, where $G(h+h',p)$  is the matrix representing the components of $g_{|p}$ in the basis in which the matrix representation of $(h_{ab}+h'_{ab})_{|p}$ is the identity.  Note that for any real matrix $A$ of rank $r$, $\|A\|_F \leq \tr(\sqrt{A^T A}) \leq \sqrt{r} \|A\|_F$, where $A^T$ is the transpose of $A$.  So 
\begin{eqnarray*}
(|g|_{h+h'})_{|p} &\leq& \tr(G(h+h',p)) = ((h^{ab}+h'^{ab})g_{ab})_{|p} \\ 
	&\leq& 2\max \{ (h^{ab}g_{ab})_{|p},(h'^{ab}g_{ab})_{|p} \} = 2\max\{ \tr(G(h,p)),\tr(G(h',p)) \} \\
	&\leq& 2\sqrt{n} \max\{ \|G(h,p)\|_F, \|G(h',p)\|_F \} = 2\sqrt{n} \max\{ (|g|_h)_{|p}, (|g|_{h'})_{|p} \}
\end{eqnarray*}
using in the second line the fact that $(\frac{1}{2}(h^{ab} + h'^{ab})\alpha_{ab})_{|p} \leq \max\{ (h^{ab}\alpha_{ab})_{|p}, (h'^{ab}\alpha_{ab})_{|p} \}$ for any $\alpha_{ab}$ at $p$.  The rightmost-hand side is bounded by hypothesis, so since $p$ was arbitrary, the conclusion follows.
\end{proof}

\begin{theorem} \label{thm:basis}
For a given $g$, the sets $B^k(g,\epsilon;h,M)$, ranging over $\epsilon > 0$ and $h \in [g]$, form a local basis at $g$.  Ranging over $g$ as well, they form a basis.
\end{theorem}
\begin{proof}
To show that the sets $B^k(g,\epsilon;h,M)$, with $(g,\epsilon;h,M) \in \Xi_G$, form a basis, it suffices to show that if $g'' \in B^k(g,\epsilon;h,M) \cap B^k(g',\epsilon';h',M)$ for arbitrary $B^k(g,\epsilon;h,M)$ and $B^k(g',\epsilon';h',M)$, then there is some other $B^k(g'',\epsilon'';h'',M)$ such that $B^k(g'',\epsilon'';h'',M) \subseteq B^k(g,\epsilon;h,M) \cap B^k(g',\epsilon';h',M)$.  So let some such $B^k(g,\epsilon;h,M)$ and $B^k(g',\epsilon';h',M)$ be given and pick any $g''$ in their intersection.  Put $h''^{ab} = h^{ab} + h'^{ab}$ and 
\[
	\epsilon'' = \min \{ \epsilon - \| g-g''\|_{h,M}, \epsilon' - \|g'-g''\|_{h',M} \}
\]
Note that $g'' \in B^k(g,\epsilon;h,M)$ by hypothesis, so lemma \ref{lemma:bounded} entails that $h \asymp g''$.  Similar reasoning gives that $h' \asymp g''$, so by proposition \ref{prop:semigroup}, $h^{ab} + h'^{ab} \in [g'']$.  Thus $B^k(g'',\epsilon'';h'',M)$ is well-defined.

Now consider any $g''' \in B^k(g'',\epsilon'';h'',M)$.  By definition, $\|g''-g'''\|_{h+h'} < \epsilon - \|g-g''\|_h$, hence
\begin{eqnarray*}
	\epsilon &>& \|g''-g'''\|_{h+h',M} + \|g-g''\|_{h,M} > \|g''-g'''\|_{h,M} + \|g-g''\|_{h,M} \\
	&\geq& \sup_M \left[ |g''-g'''|_h + |g-g''|_h \right] \geq \|g-g'''\|_{h,M},
\end{eqnarray*}
using the fact that $\|g''-g'''\|_{h+h',M} > \|g''-g'''\|_{h,M}$ in the first line and the triangle inequality for the supremum and then the $h$-fiber norm in the second.  Therefore $g''' \in B^k(g,\epsilon;h,M)$ and a similar argument shows that $g''' \in B^k(g',\epsilon';h',M)$.  Hence $B^k(g'',\epsilon'';h'',M) \subseteq B^k(g,\epsilon;h,M) \cap B^k(g',\epsilon';h',M)$ since $g'''$ was arbitrary.

To show that the sets $B^k(g,\epsilon;h,M)$, ranging only over $\epsilon$ and $h \in [g]$, form a local basis at $g$, it suffices to show that, for every basic open neighborhood $B^k(g',\epsilon';h',M) \ni g$ such that $h' \asymp g'$, there is some $B^k(g,\epsilon;h,M) \subseteq B^k(g',\epsilon';h,M)$ such that $h \asymp g$.  So consider some particular such $B^k(g',\epsilon';h',M)$ and put $h = h'$ and $\epsilon = \epsilon' - \|g-g'\|_{h',M}$.  By lemma \ref{lemma:bounded}, $h \asymp g'$.  Now, for an arbitrary $g'' \in B^k(g,\epsilon;h,M)$, $\|g-g''\|_{h,M} < \epsilon' - \|g-g'\|_{h',M}$, so applying similar reasoning as above yields that $\|g-g''\|_{h',M} < \epsilon'$, i.e., $g'' \in B^k(g',\epsilon';h',M)$, hence $B^k(g,\epsilon;h,M) \subseteq B^k(g',\epsilon';h',M)$.
\end{proof}

Now, the definitional subbases for the compact-open and open topologies, $\Xi_{CO}$ and $\Xi_O$, respectively, are plainly closely related---the only difference between them is in their fourth component.  Despite initial appearances, there is a definitional subbasis $\Xi'_{CO}$ generating the compact-open topology that is plainly closely related to $\Xi_G$, also differing only in its fourth component.

\begin{proposition}
The definitional subbasis $\Xi'_{CO} = \{ (g,\epsilon;h,C) : g \in L(M), \epsilon \in (0,\infty), C \in \mathcal{C}, h_{|C} \asymp g_{|C} \}$ generates the compact-open topology; in fact, $\Xi'_{CO} = \Xi_{CO}$.
\end{proposition}
\begin{proof}
Let any $g \in L(M)$, $C \in \mathcal{C}$, $K \in \Gamma^r_s$, and $h \in \mathcal{R}(M)$ be given.  Since $C$ is compact, the continuous scalar fields $|K|_f$ and $|K|_h$ are both bounded on $C$ \cite[Theorem 17.13, p.~123]{Willard70}, hence $\mathcal{K}^r_s(f,C)=\Gamma^r_s=\mathcal{K}^r_s(h,C)$ and by definition $\mathcal{R}(M) \subseteq [g_{|C}]$.
\end{proof}
Thus the global topologies can be understood as a natural variation on the compact-open topologies that controls similarity across $M$ in just the same way as the open topologies.

\subsection{Properties of the Global Topologies}

Note that the global topologies are indeed topologies for $L(M)$ since $[g] \cap \mathcal{R}(M)$ is non-empty for each $g$.  In particular, for any $h \in \mathcal{R}(M)$ and $g \in L(M)$, $h/|g|_h \asymp g$. Further, they are \emph{diffeomorphism-invariant} \cite[p.~281--2]{Geroch70}, in the sense that for any element $\psi$ of the diffeomorphism group of $M$, the pushforward $\psi_*$ acts as a homeomorphism on the Lorentz metrics equipped with a global topology.  To see this, first note that
\begin{eqnarray*}
|\psi^* g|_{\psi^* h} &=&  [\psi^*(h^{am})\psi^*(h^{bn})\psi^*(g_{ab})\psi^*(g_{mn})]^{1/2} \\ &=& [\psi^*(h^{am}h^{bn}g_{ab}g_{mn})]^{1/2} = [h^{am}h^{bn}g_{ab}g_{mn}]^{1/2} = |g|_h,
\end{eqnarray*}
so $h \asymp g$ if and only if $\psi^*h \asymp \psi^*g$.  Using similar reasoning, $\psi^*[B^k(g,\epsilon;h,M)]=B^k(\psi^*g,\epsilon;\psi^*h,M)$.  Thus the preimage of every basis element is open, meaning $\psi_*$ is continuous, hence acts as a homeomorphism.

The $C^0$ global topology rules in the ``right'' way on Geroch's three examples.  First, the sequence defined by eq.~\ref{eq:example1a} does not converge to the Minkowski metric $\eta$.  To see why, consider
\[
	h^{ab} = \left( \frac{\partial}{\partial t} \right)^a \left( \frac{\partial}{\partial t} \right)^b + 
	\left( \frac{\partial}{\partial x} \right)^a \left( \frac{\partial}{\partial x} \right)^b + 
	\left( \frac{\partial}{\partial y} \right)^a \left( \frac{\partial}{\partial y} \right)^b + 
	\left( \frac{\partial}{\partial z} \right)^a \left( \frac{\partial}{\partial z} \right)^b,
\]
noting that $\|\eta\|_h = 2$, so $h \asymp \eta$.  This choice yields that
\[
	\|\eta-\stackrel{m}{g}\|_h = \sup_M \frac{m}{1+(x-m)^2} = m,
\]
which cannot be as small as one wishes for sufficiently large $m$.  Hence $\stackrel{m}{g}$ does not converge to the Minkowski metric.

Second, the sequence defined by eq.~\ref{eq:example2} \emph{does} converge to the Minkowski metric.  To see why, let 
\begin{equation}
h^{ab} = \alpha  \left( \frac{\partial}{\partial t} \right)^a \left( \frac{\partial}{\partial t} \right)^b + \beta_1  \left( \frac{\partial}{\partial x} \right)^a \left( \frac{\partial}{\partial x} \right)^b + \beta_2  \left( \frac{\partial}{\partial y} \right)^a \left( \frac{\partial}{\partial y} \right)^b + \beta_3  \left( \frac{\partial}{\partial z} \right)^a \left( \frac{\partial}{\partial z} \right)^b + \cdots
\end{equation}
be any Riemannian metric in $[\eta]$ in coordinates determined by $\eta$.  (Note that $h$ need not be diagonal in these coordinates.)  Since $\eta \in \mathcal{K}^2_0(\eta,M)$, by definition $\eta \in \mathcal{K}^2_0(h,M)$, hence
\begin{equation}
\infty > \|\eta\|_{h,M} = \sup_M (h^{ab}h^{cd}\eta_{ac}\eta_{bd})^{1/2} = \sup_M (\alpha^2 + \beta_1^2 + \beta_2^2 + \beta_3^2)^{1/2} > \sup_M |\alpha|.
\end{equation}
Consequently
\[
	\|\eta - \stackrel{m}{g}{}'\|_{h,M} = \sup_M \frac{|\alpha|}{m^2 + x^2 + y^2 + z^2} \leq \frac{\sup_M |\alpha|}{m^2},
\]
whose right-hand side is finite and can be made as small as one wishes by choosing a sufficiently large $m$.  Hence $\stackrel{m}{g}{}' \to \eta$ as $m \to \infty$.  Third, and finally, by similar calculations one can show that the family $\{ \lambda g_{ab} : \lambda > 0 \}$ is continuous in the $C^0$ global topology.

We can now use these facts and theorem \ref{thm:basis} to show that each $C^k$ global topology lies strictly between the $C^k$ compact-open and $C^k$ open topologies for non-compact manifolds.  (Since the compact-open and open topologies coincide when $M$ is compact, clearly the global topologies do as well.)

\begin{proposition} \label{prop:compare}
For $M$ non-compact, the $C^k$ global topology on $L(M)$ is strictly finer that the $C^k$ compact-open topology and strictly coarser than the $C^k$ open topology.
\end{proposition}
\begin{proof}
Since $\Xi_G \subset \Xi_O$, each basic neighborhood of the $C^k$ global topology is a basic neighborhood of the $C^k$ open topology, but
the former makes every family of the form $\{ \lambda g_{ab} : \lambda > 0 \}$ is continuous, while the latter does not.  Thus the $C^k$ global topology is strictly coarser than the $C^k$ open topology.

Next let an arbitrary $B(g,\epsilon;h,C)$ be given, and pick any $h' \in [g] \cap \mathcal{R}(M)$.  Define $\Omega = |g|_{h'}/|g|_h$ and put $h''^{ab} = \Omega h^{ab} / \inf_C \Omega$.  Since
\[
	\|g\|_{h'',M} = \sup_M (h''^{am}h''^{bn}g_{ab}g_{mn})^{1/2} = \frac{\|g\|_{h',M}}{\inf_C \Omega} < \infty,
\]
we have that $h'' \asymp g$, so $B^k(g,\epsilon;h'',M)$ is a basic neighborhood of the $C^k$ global topology.  Now, for any $g' \in B^k(g,\epsilon;h'',M)$,
\[
	\epsilon > \|g-g'\|_{h'',M} \geq \|g-g'\|_{h'',C} = \frac{\sup_C (\Omega |g-g'|_h)}{\inf_C \Omega} \geq \left(\frac{\sup_C \Omega}{\inf_C \Omega}\right) \sup_C |g-g'|_h \geq \|g-g'\|_{h,C},
\]
thus $g' \in B(g,\epsilon;h,C)$.  Since $g'$ was arbitrary, we have $B^k(g,\epsilon;h'',M) \subseteq B(g,\epsilon;h,C)$, so by the Hausdorff criterion \cite[Theorem 4.8, p.~35]{Willard70}, the $C^k$ compact-open topology is coarser than the $C^k$ global topology.  

Lastly, for the sake of contradiction, suppose that every basic neighborhood of the $C^k$ global topology $B^k(g,\epsilon;h,M)$ contains a basic neighborhood of the $C^k$ compact-open topology, $B^k(g,\epsilon';h',C)$.  Consider $g' = (1+ \rho/|g|_h)g$, where $\rho$ is a smooth positive scalar field such that
\[
\sup_C \rho < \epsilon' \left( \sup_C \frac{|g|_{h'}}{|g|_h} \right)^{-1},
\]
but $\sup_M \rho$ does not exist, i.e., is infinite.  Note that
\[
\|g-g'\|_{h',C} = \sup_C \rho \frac{|g|_{h'}}{|g|_h} \leq (\sup_C \rho) \left(\sup_C \frac{|g|_{h'}}{|g|_h}\right) < \epsilon',
\]
so $g' \in B^k(g,\epsilon';h',C)$.  But $|g-g'|_h = \rho$, which is unbounded on $M$, hence a contradiction. Thus by the Hausdorff criterion the $C^k$ global topology is not coarser than the $C^k$ compact-open topology.
\end{proof}

Because the global topologies are generated from a collection of norms, the corresponding collection of metrics (defined by $d(g,g';h)=\|g-g'\|_h$) in fact defines a uniform structure on $L(M)$.  Instead of digressing into the details of the theory of uniform spaces---that is, sets endowed with uniform structure---the following definition relevant to the present investigation may be abstracted from that theory.

\begin{definition}
A pair of points $x,x'$ in a space $X$ whose topology is generated from the $\epsilon$-balls $\mathcal{B}(x,\epsilon;d_\alpha)$ of a collection of metrics $\{d_\alpha\}$ is said to be \emph{uniformly connected} when for every $\epsilon > 0$ and metric $d_\alpha$, there is a finite sequence $\mathcal{B}(y_1,\epsilon;d_\alpha), \ldots, \mathcal{B}(y_n,\epsilon;d_\alpha)$ such that $x \in \mathcal{B}(y_1,\epsilon;d_\alpha)$, $x' \in \mathcal{B}(y_n,\epsilon;d_\alpha)$, and $\mathcal{B}(y_i,\epsilon;d_\alpha) \cap \mathcal{B}(y_j,\epsilon;d_\alpha) \neq \emptyset$ if $|i-j|=1$.  The space $X$ is uniformly connected when each pair of its points is uniformly connected. (Sometimes this property is called uniform chain connectedness \cite{MP64}, in analogy with the chain connectedness for general topological spaces, although this latter property is strictly weaker than (topological) connectedness \cite[Theorem 26.15, p.~195]{Willard70}.)
\end{definition}
Such a space is \textit{locally uniformly connected} when every point has a local neighborhood basis consisting of sets that are uniformly connected (in the subspace topology).

\begin{proposition} \label{prop:connected}
Each $B^k(g,\epsilon;h,M)$ (with $h \asymp g$) is uniformly connected in the $C^k$ global topology.
\end{proposition}
\begin{proof}
Let some such $B^k(g,\epsilon;h,M)$ be given and consider some arbitrary $h' \asymp g$, $\epsilon' > 0$, and $g' \in B^k(g,\epsilon;h,M)$.  Since then $\|g - g' \|_{h'} < \infty$, pick some positive $c < \epsilon'/\|g - g' \|_{h'}$, so that $\|cg - cg' \|_{h'} < \epsilon'$.  Further, note that $\|g - cg\|_{h'} = |1-c|\cdot \|g\|_{h'}$ and $\|g' - cg'\|_{h'} = |1-c|\cdot \|g'\|_{h'}$ are both finite, so put $N = \lfloor 1 + \|g - cg\|_{h'}/\epsilon' \rfloor$, $N' = \lfloor 1 + \|g' - cg'\|_{h'}/\epsilon' \rfloor$, $C=\|g - cg\|_{h'}/N$, and $C'=\|g' - cg'\|_{h'}/N'$.  The families $\Lambda = \{g - \lambda g: \lambda \in [0,c] \}$ and $\Lambda' = \{g' - \lambda g': \lambda \in [0,c] \}$ can then be covered by neighborhoods of the form $B^k(g-nCg,\epsilon';h',M)$ and $B^k(g'-n'C'g',\epsilon';h',M)$, respectively, with $n \in \{0, \ldots, N\}$ and $n' \in \{0, \ldots, N'\}$.  This shows that $g$ and $g'$ are uniformly connected, but since $g'$ was arbitrary and uniform connection is an equivalence relation, each pair of elements in $B^k(g,\epsilon;h,M)$ is uniformly connected.
\end{proof}
Combining propositions \ref{prop:connected} and \ref{prop:compare} immediately yields the following. 
\begin{corollary}
Each $C^k$ global topology is locally uniformly connected.
\end{corollary}
Despite being locally uniformly connected, the global topologies are not uniformly connected.  In fact, they have uncountably many uniform components that can be described using the following lemma.

\begin{lemma} \label{lemma:connection}
Given $g,g' \in L(M)$, $g \asymp g'$ if and only if there are constants $\epsilon,\epsilon'>0$ and $h,h' \in \mathcal{R}(M)$ such that $h \asymp g$, $h' \asymp g'$, and $B^k(g,\epsilon;h,M) \cap B^k(g',\epsilon';h',M) \neq \emptyset$.
\end{lemma}
\begin{proof}
Suppose that $g \asymp g'$.  Pick any $h \in [g] \cap \mathcal{R}(M)$, noting that $h \asymp g'$ by transitivity, and then pick any $\epsilon > \|g-g'\|_{h,M}$.  Hence $g' \in B^k(g,\epsilon;h,M)$, so arbitrarily setting $\epsilon' = \epsilon$ and $h'=h$ yields that $B^k(g,\epsilon;h,M) \cap B^k(g',\epsilon';h',M) \neq \emptyset$.

Conversely, suppose that $h \asymp g$, $h' \asymp g'$, and $g'' \in B^k(g,\epsilon;h,M) \cap B^k(g',\epsilon';h',M)$.  By lemma \ref{lemma:bounded}, $h \asymp g''$ and $h' \asymp g''$, hence by transitivity $g \asymp g'$.
\end{proof}

\begin{theorem}
The uniformly connected components of the $C^k$ global topology on $L(M)$ have the following properties:
\begin{enumerate}
\item they are identical with the norm-equivalence classes $[f] \cap L(M)$;
\item through translation each such component (under the subspace topology) generates a locally convex topological vector space on $\Gamma^0_2$ (and similarly on $\Gamma^2_0$ for the inverse metrics) compatible with the original topology on the component, in the sense that taking the subspace topology again returns the original topology.
\end{enumerate}
\end{theorem}
\begin{proof}
\begin{enumerate}
\item Let $\mathcal{G}(g) = \bigcup_{ \epsilon > 0} B^k(g,\epsilon;h,M)$ for each $g \in L(M)$.  Now, by lemma \ref{lemma:connection}, $g' \nasymp g$ if and only if $\mathcal{G}(g) \cap \mathcal{G}(g') = \emptyset$.  So $\mathcal{G}(g)=[g]$.  Further, 
since any pair of elements of $\mathcal{G}(g)$ is in $B^k(g,\epsilon;h,M)$ for some finite $\epsilon > 0$, proposition \ref{prop:connected} entails that $\mathcal{G}^k(g)$ is uniformly connected. Thus they are the maximal uniformly connected components.

\item Since the Lorentz metrics span $\Gamma^0_2$, a topology on the latter is generated from each uniform component as the final topology induced from the translation maps.  Moreover, these are generated from the $(g,M)$-uniform norms, so by proposition \ref{prop:tvs} the resulting locally convex topology is compatible with the linear structure of $\Gamma^0_2$.  The local bases for each $g$ in the same component are generated from the same collections of norms, so the translation maps add no new open neighborhoods to the Lorentz metrics (modulo elements that are not Lorentz metrics).
\end{enumerate}
\end{proof}
Although the topology that each uniform component generates is locally convex, it bears remarking that the collection $L(M)$ is not itself a convex set, for in general the set $\{ \lambda g_{ab} + (1-\lambda)g'_{ab} : \lambda \in [0,1] \} \nsubset L(M)$ for arbitrary $g,g' \in L(M)$.

Lastly, one can characterize the global topologies in terms of similarity of observable quantities, the fields definable on $M$ in terms of $g$ and any collection of frame fields whose $(g,M)$-uniform norm is bounded.

\begin{proposition} \label{prop:global interpretation}
A family of tensor fields $\stackrel{\lambda}{\phi}{}^a_{bc}$ on corresponding spacetimes $(M,\stackrel{\lambda}{g})$, with $\lambda \in (0,a)$ for some $a>0$, converges to a tensor field $\phi^a_{bc}$ on a spacetime $(M,g)$ in the $C^k$ global topology iff $\lim_{\lambda \to 0} \sup_M (\stackrel{\lambda}{\phi}{}^a_{bc} \stackrel{0}{\psi}{}^{bc}_a) = \sup_M \phi^a_{bc} \stackrel{0}{\psi}{}^{bc}_a$ for every tensor field $\stackrel{0}{\psi}{}^{bc}_a \in \mathcal{K}^2_1(g,M)$, and for each positive $j \leq k$, $\lim_{\lambda \to 0} \sup_M (\stackrel{j}{\psi}{}^{bcd_1 \cdots d_j}_a \nabla_{d_j} \cdots \nabla_{d_1} \stackrel{\lambda}{\phi}{}^a_{bc}) = \sup_M \stackrel{j}{\psi}{}^{bcd_1 \cdots d_j}_a \nabla_{d_j} \cdots \nabla_{d_1}\phi^a_{bc}$ for every tensor field $\stackrel{j}{\psi}{}^{bcd_1 \cdots d_j}_a$. Moreover, the $C^k$ global topology is the unique topology with this property.
\end{proposition}
\begin{proof}
Analogous to that of \cite[Prop.~3.3]{Fletcher14a}.
\end{proof}
Analogous propositions hold for tensor fields $\phi$ of other ranks.  One can interpret any particular $\psi$ field as determining a kind of local system of rods and clocks by which the $\phi$ fields are measured.  The restriction on the former amounts to the requirement that they do not get arbitrarily long and rapid, corresponding to arbitrary precision at infinity, percentage-wise.  In a word, then, a field converges in the $C^k$ global topology just in case all observers with bounded precision agree that their measurements converge to those they would make of the limit field.  It is instructive to contrast this description with that for the $C^k$ open topologies:
\begin{proposition} \label{prop:open interpretation}
A family of tensor fields $\stackrel{\lambda}{\phi}{}^a_{bc}$ on corresponding spacetimes $(M,\stackrel{\lambda}{g})$, with $\lambda \in (0,a)$ for some $a>0$, converges to a tensor field $\phi^a_{bc}$ on a spacetime $(M,g)$ in the $C^k$ open topology iff there is some compact $C \subseteq M$ such that $(\stackrel{\lambda}{\phi}{}^a_{bc})_{|M \setminus C} = (\phi^a_{bc})_{|M \setminus C}$ for sufficiently small $\lambda$, $\lim_{\lambda \to 0} \sup_C (\stackrel{\lambda}{\phi}{}^a_{bc} \stackrel{0}{\psi}{}^{bc}_a) = \sup_C \phi^a_{bc} \stackrel{0}{\psi}{}^{bc}_a$ for every tensor field $\stackrel{0}{\psi}{}^{bc}_a \in \mathcal{K}^2_1(g)$, and for each positive $j \leq k$, $\lim_{\lambda \to 0} \sup_M (\stackrel{j}{\psi}{}^{bcd_1 \cdots d_j}_a \nabla_{d_j} \cdots \nabla_{d_1} \stackrel{\lambda}{\phi}{}^a_{bc}) = \sup_C \stackrel{j}{\psi}{}^{bcd_1 \cdots d_j}_a \nabla_{d_j} \cdots \nabla_{d_1}\phi^a_{bc}$ for every tensor field $\stackrel{j}{\psi}{}^{bcd_1 \cdots d_j}_a$.
\end{proposition}
\begin{proof}
Apply proposition \ref{prop:open convergence} and the proof from \cite[Prop.~3.2]{Fletcher14a}.
\end{proof}
Again, analogous propositions hold for tensor fields $\phi$ of other ranks.  In contrast to proposition \ref{prop:global interpretation}, the $\psi$ fields by which the convergence of $\phi$ is determined are unrestricted, but in balance with that weakening a much stronger condition is placed on the suprema, namely that they equal the limit point except for a bounded region.  This makes sense, for if observers are allowed to have arbitrary precision at infinity, the only way a field can converge is if there is no difference between it and its limit point outside of a bounded region.

\section{Conclusions and Prospects}\label{sec:conclusion}

Geroch has remarked that finding an appropriate topology is surprisingly difficult, and at times has expressed some measure of doubt regarding whether there even \textit{is} such a topology.  Instead of reading his demand as one for a canonical topology, one can see him searching for a particular topology with certain properties.  The construction of the global topologies meet that goal: it is sensitive to globally defined properties but retains enough continuity for the linear operations defined on $L(M)$ to capture one's intuitive judgments regarding the convergence and continuity of his examples.

A natural query to then pose regards whether, and which, theorems in the literature that use the open topologies have analogs using their global topology counterparts that are more physically relevant in light of the problems with the former.  Much work remains to be done for this query, but it turns out that at least certain types of theorems carry over exactly.  Recall that a property $P$ is \textit{conformally invariant} when $P$ holds of $g$ if and only if for every scalar field $\Omega > 0$, $P$ holds for $\Omega g$.
\begin{theorem} \label{thm:correspondence}
A conformally invariant property $P$ of a spacetime $g$ is stable (resp. dense) on some $S \subseteq L(M)$ in the $C^k$ global topology if and only if it is stable (resp. dense) on $S$ in the $C^k$ open topology.
\end{theorem}
\begin{proof}
For stability, it suffices to show the above equivalence holds when $S=g$, a point.  One direction follows immediately from proposition \ref{prop:compare}, so suppose that some conformally invariant property of $g$ is stable in the $C^k$ open topology, i.e., there is some $B^k(g,\epsilon;h,M)$ all of whose elements have that property.  Let $\Omega = 1/|g|_h$, one can express $B^k(g,\epsilon;h,M) = \{ g' : \|\Omega^{-1} g - \Omega^{-1} g'\|_{\Omega h} < \epsilon \}$.  Because $P$ is conformally invariant, it must hold on the set $B^k(\Omega g,\epsilon;h,M) = \{ g' : \|g - g'\|_{\Omega h} < \epsilon \} = B^k(g,\epsilon;\Omega h,M)$, which is a basic open neighborhood of $g$ in the $C^k$ global topology since $|g|_{\Omega h} = 1$ implies that $g \asymp \Omega h$.

For denseness, as before, one direction follows immediately from proposition \ref{prop:compare}, so suppose that some conformally invariant property $P$ of $g$ is dense in some $S \subseteq L(M)$ in the $C^k$ global topology, i.e., for every $g \in S$ there is some $B^k(g,\epsilon;h,M)$ with $h \asymp g$ that contains an element $g'$ with $P$.  Now consider any basic neighborhood $B^k(g,\epsilon';h',M)$ in the $C^k$ open topology.  Using the same reasoning as above, one can conclude that it contains an element with property $P$ if $B^k(\Omega' g, \epsilon';h',M) = B^k(g,\epsilon'; \Omega' h', M)$ does so as well for some scalar field $\Omega' > 0$.  Picking $\Omega' = \max\{\epsilon''/2|g|_{h'},\epsilon''/2|g'|_{h'}\}$, where $\epsilon'' \in (0,\epsilon')$, yields that
\[
|g-g'|_{\Omega' h'} \leq |g|_{\Omega' h'} + |g'|_{\Omega' h'} = \Omega|g|_{h'} + \Omega|g'|_{h'} \leq \epsilon''/2 + \epsilon''/2 < \epsilon'.
\]
Hence $\|g-g'\|_{\Omega' h'} < \epsilon'$, meaning $g' \in B^k(g,\epsilon'; \Omega' h', M)$.
\end{proof}
Recall that Hawking \cite{Hawking69} showed that the existence of a global time function---a smooth scalar field strictly increasing on each future-directed timelike curve---is equivalent to stable causality, i.e., the stability of the property of having no closed causal curves in the $C^0$ open topology. (See also \cite[Propopsition 6.4.9, p. 198--201]{HE73}.  These early results actually only proved that the existence of a \textit{continuous} time function follows from stable causality; the extension to smoothness remained a folk theorem until surprisingly recently.  See \cite[\S3.8.3]{MS08} and references therein for part of this story.)  We thus have the following:
\begin{corollary}
Hawking's theorem holds for the $C^0$ global topology, i.e., a spacetime admits of a global time function if and only if it is stable in the $C^0$ global topology.
\end{corollary}
There are many properties, such as being singular, not covered by theorem \ref{thm:correspondence}.  Geodesic (in)completeness, for instance, is sometimes but not always stable \cite[Ch. 7.1]{BEE96}, and I suspect one can find examples thereof for which the open and global topologies render different judgments.  Again, in light of the open topologies' many problems, the global topologies seem a better choice for formulating these kinds of questions regarding global properties.  

Of course, I have not shown that the global topologies are the unique topologies meeting Geroch's desiderata.  While I have given several characterizations of their structure, proving that they have further invariant characteristics, such as being ``maximal" or ``minimal" in some relevant way, would further illuminate their status.  For instance, one might try to make precise and then investigate a sense in which the global topologies might make linear operations defined on $L(M)$ ``as continuous as they could be" (perhaps subject to some other constraints).

Another related direction to pursue concerns the fact that both the compact-open and open topologies have natural formulations in terms of fiber bundle theory.  Recall that each Lorentz metric corresponds to a smooth cross-section of a $(0,2)$-tensor bundle over $M$, and its derivatives to cross-sections of the appropriate jet bundle.  Given an open set of the total space of the bundle, one can define a basis element for the open topology as the set of Lorentz metrics whose corresponding cross-sections' images lie in that open set. One can define a subbasis element for the compact-open topology similarly except one considers the cross-sections' images restricted to compacta of $M$.  I suspect that the global topologies can be given a natural fiber bundle formulation.  Like in the present investigation, this would require attention to the algebraic structure of the sections, and so may require something like the principle bundle formalism.

\section*{Acknowledgements}
Thanks to Jim Weatherall and Chris W\"uthrich for encouraging comments, and to audience members in Pittsburgh, Pasadena, Cambridge, and Nijmegen for illuminating discussion. 
Part of the research leading to the work on which this manuscript was based \cite[Ch.~4]{Fletcher14a} was completed with the support of a National Science Foundation Graduate Research Fellowship.

%
\bibliographystyle{unsrt}
%

\end{document}